\documentclass{article}
\usepackage{graphicx,amssymb,amsmath,amsthm}
\usepackage[usenames,dvipsnames]{xcolor}

\newtheorem{theorem}{Theorem}[section]

\newtheorem{lemma}[theorem]{Lemma}

\newtheorem{problem}{Problem}

\newtheorem{defn}{Definition}

\newcommand{\MyQuote}[1]{\vspace{0.5cm}%
     \parbox{10cm}{\em #1}\hspace*{2cm}($\ast$)\\[0.5cm]}

\title{Drawing the Horton Set in an Integer Grid of Minimum Size}

\author{Luis Barba \thanks{D\'epartement d'Informatique, ULB, Belgium.} \thanks{School of Computer Science, Carleton University, Canada.
} \and
  Frank Duque \thanks{Departamento de Matem\'aticas, CINVESTAV. Partially supported by Conacyt of
Mexico, Grant 153984.} \and
 Ruy Fabila-Monroy \footnotemark[3] \and
 Carlos Hidalgo-Toscano \footnotemark[3]}

\begin{document}
\maketitle

\begin{abstract}

In 1978 Erd\H os asked if every sufficiently large set of points in general
position in the plane contains the vertices of a convex $k$-gon, with the additional
property that no other point of the set lies in its interior. Shortly after, Horton provided a construction---which is now called
the Horton set---with no such $7$-gon. In this paper we show that the Horton set of $n$ points can be realized with 
integer coordinates of absolute value at most $\frac{1}{2} n^{\frac{1}{2} \log (n/2)}$.
We also show that any set of points with integer coordinates combinatorially equivalent (with the same order type)
to the Horton set,
 contains a point with  a coordinate of absolute value at least 
 $c \cdot n^{\frac{1}{24}\log (n/2)}$, where $c$ is a positive constant.

\end{abstract}

\section{Introduction}
Although the number of distinct sets of $n$ points in the plane is infinite,
for most problems in Combinatorial Geometry only a 
finite number of them can be considered as essentially distinct. 
Various equivalence relations on point sets have been proposed
by Goodman and Pollack~\cite{sec,semispaces,complexity,m_sorting}. 
One of these is the \emph{order type}; it is defined
on a set of points, $S$, as follows. To every triple $(p,q,r)$ of points of $S$, assign: a $-1$ if $r$
lies to the left of the oriented line from $p$ to $q$; a $1$ if $r$ lies
to the right of this line; and a $0$ if $p$, $q$, and $r$ are collinear.
This assignment is also called the orientation of the triple, which may
be \emph{negative}, \emph{positive} or \emph{zero}, respectively. The set of all
triples together with their orientation is the order type of $S$; 
two sets of points have the same
order type if there is a bijection between them that preserves orientations.

Since their inception, order types were defined with computational
applications in mind (see~\cite{m_sorting} for example). The orientation of a triple is
determined by the sign of a determinant; many algorithms
use precisely this determinant as their geometric primitive. 
Given that the determinant of an integer valued matrix is an integer,
for numerical computations it is best if a point set has
integer coordinates. Two main reasons are that integer arithmetic is much
faster than floating point arithmetic, and that floating point
arithmetic is prone to rounding errors. The latter is easily taken care of
 with an integer representation that can handle arbitrarily large numbers.

If a set of $n$ points has already integer coordinates, it is best if these coordinates have as small
 absolute value as possible---again,
for computational reasons. Even though rounding errors can be avoided using arbitrarily large
integers, the cost of computation
increases as the numbers get larger. Also, if we wish
to store the point set, the number of bits needed depend on the
size of the coordinates.

Let $S$ be a set of $n$ points in general position in the plane.
A \emph{drawing} of $S$ is a set of points with integer
coordinates and with the same order type as $S$. The
\emph{size} of a drawing is the maximum of
the absolute values of its coordinates. For the reasons mentioned above,
 it is of interest to find the drawing of $S$ of minimum size.
In~\cite{goodman}  Goodman, Pollack and Sturmfels
presented sets of $n$ points in general position whose
smallest drawings have size $2^{2^{c_1 n}}$, and proved  that every point
set has a drawing of size at most $2^{2^{c_2 n}}$ (where
$c_1$ and $c_2$ are positive constants). 

Aichholzer, Aurenhammer and Krasser~\cite{database} have assembled a database of drawings.
For $n=3, \dots, 11$, the database contains a drawing of every possible set
of $n$ points in general position in the plane.
The main advantage of having these drawings is that one can use them
to compute certain combinatorial parameters of all point sets up to eleven points. 
The order type data base stops at eleven because the size of the database
grows prohibitively fast. Thus, we cannot hope to store drawings for all point sets
beyond small values of $n$; it is convenient however, to have programs that generate
small drawings of infinite families of point sets which
are of known interest in Combinatorial Geometry. 

In this direction, Bereg et al. \cite{double_circle} provided a linear time algorithm to generate
a drawing of a set of points called the Double Circle\footnote{The Double
 Circle of $2n$ points is constructed as follows. Start with
a convex $n$-gon; arbitrarily close to the midpoint 
of each edge, place a point in the interior of this polygon;
finally place a point at each vertex of the polygon.}.
Their drawing has size $O(n^{3/2})$; they also proved a lower bound of $\Omega(n^{3/2})$ on the size of every drawing
of the Double Circle.  In this paper we do likewise for a point set called the Horton Set~\cite{hortonsets}.
In section~\ref{sec:upper} we provide a drawing of size $\frac{1}{2} n^{\frac{1}{2} \log (n/2)}$ of the Horton set of $n$ points;
our drawing can be easily constructed in linear time. We also show in Section~\ref{sec:lower},  
a lower bound of $c \cdot n^{\frac{1}{24}\log (n/2)}$ (for some $c>0$) on the minimum size 
of any drawing of the Horton set. As a corollary, $\Theta(n\log^2 n)$ bits
are necessary and sufficient to store a drawing of the Horton set. 

We are mainly interested in having an algorithm that generates small drawings of the Horton set. 
However, the problem of finding small drawings 
also raises interesting theoretical questions. For example, after learning of our
lower bound, Alfredo Hubard posed the following problem.

\begin{problem}
Does every sufficiently large set of points, for which there exist a drawing of polynomial size, contains 
an empty $7$-hole?
\end{problem}

A \emph{$k$-hole} of a point set $S$, is a  subset of $k$ points of $S$ that form
a convex polygon, with no other point of $S$ in its interior. 
Horton sets were constructed as an example of arbitrarily large point sets without 
$7$-holes. In particular our lower bound implies that any set of points 
that has a drawing of polynomial size, cannot
have large copies of the Horton set.
We also believe that the machinery developed to prove Theorem~\ref{thm:lower_gen} will be
useful for analyzing Horton sets in other settings.

A preliminary version of this paper appeared in CCCG'14~\cite{HortonCCCG}. In this paper
all point sets are in general position and all logarithms are base 2.
\subsection{The Horton Set(s)}\label{sec:horton}

In 1978 Erd{\H o}s~\cite{somemore} asked if for every $k\ge 3$, 
any sufficiently large set of points in the plane
contains a $k$-hole. 
Shortly after, Harborth~\cite{harborth} showed that every set of $10$ points contains
a $5$-hole. The case of empty triangles ($3$-holes) is trivial;  
the case of $4$-holes was settled in the affirmative in another context
by Esther Klein long before Erd\H os posed his question (see \cite{happyend}).
Horton~\cite{hortonsets} constructed arbitrarily large point sets without $7$-holes,
and thus without $k$-holes for larger values of $k$. His construction is now known as
the Horton set. The case of $6$-holes remained open for almost 30 years, until
Nicol\'as \cite{nicolas}, and independently Gerken \cite{gerken}, proved
that every sufficiently large set of points contains a $6$-hole.

Since its introduction, the Horton set has been used as an extremal
example in various combinatorial problems on point sets. 
For example, a natural question is to ask: What is the minimum number
of $k$-holes in every set of $n$ points in the plane? The case
of empty triangles was first considered by Katchalski and Meir~\cite{katmeir}---they constructed 
a set of $n$ points with $200n^2$ empty triangles
and showed that every set of $n$ points contains $\Omega(n^2)$ of them. 
This bound  was later improved by B\'ar\'any and F\"uredi~\cite{emptysimplices}, who
 showed that the Horton set has $2n^2$
empty triangles. The Horton set was then used in a series of papers as a building block
to construct sets with fewer and fewer $k$-holes. The first of these constructions was given by Valtr~\cite{valtr95};
it was later improved by Dumitrescu~\cite{dumi00} and the final improvement was given by B\'ar\'any and Valtr~\cite{valtr04}.

Devillers et al.~\cite{chromaticvariants} considered chromatic variants of these problems. In particular,
they described a three-coloring of the points of the Horton set with no empty monochromatic triangles. Since every set
of $10$ points contains a $5$-hole, every two-colored set of at least $10$ points contains
 an empty monochromatic triangle. The first non trivial lower bound of $\Omega(n^{5/4})$, on 
the number of empty monochromatic triangles
in every two-colored set of $n$ points, was given by Aichholzer et al~\cite{trimono}. 
This was later improved by Pach and T\'oth~\cite{tripach} to $\Omega(n^{4/3})$. 
The known set with the least number of empty monochromatic triangles is 
given in \cite{trimono}; it is based on the known set with the fewest
number of empty triangles, which in turn is based on the Horton set.

We now define the Horton set.
Let $S$ be a set of $n$ points in the plane with no two points
having the same $x$-coordinate; sort
its points by their $x$-coordinate so that 
$S:=\{p_0, p_1,\dots, p_{n-1}\}$. Let $S_{\textrm{even}}$ be the subset of the
even-indexed points, and  $S_{\textrm{odd}}$ be the subset of the 
odd-indexed points. That is, $S_{\textrm{even}}=\{p_0, p_2,\dots\}$
and $S_{\textrm{odd}}=\{p_1, p_3,\dots\}$. Let $X$ and $Y$ be two sets of points in the plane.
We say that $X$ is \emph{high above} $Y$ if: every line determined by two points in $X$ is above every point in $Y$, and
every line determined by two points in $Y$ is below every point in $X$.

\begin{defn} \label{def:mat} 
The \textbf{Horton set} is a set $H^k$ of $2^k$ points, with no two points having
the same $x$-coordinate, that satisfies the following properties.
\begin{enumerate}
  \item $H^0$ is a Horton set;

  \item both $H_{\textrm{even}}^k$ and $H_{\textrm{odd}}^k$ are Horton sets ($k \ge 1$);
  
  \item $H_{\textrm{odd}}^k$ is high above $H_{\textrm{even}}^k$ ($k \ge 1$).
\end{enumerate}
\end{defn}

This definition is very similar to the one given in Matou{\v{s}}ek's book~\cite{mat} (page 36).
The only difference is that in that definition either
$H_{\textrm{even}}^k$ is high above $H_{\textrm{odd}}^k$ or $H_{\textrm{odd}}^k$ is high above $H_{\textrm{even}}^k$; i.e. this relationship
is allowed to change at each step of the recursion. As a result, for
a fixed value of $k$, one gets a family of 
``Horton sets'' (with different order types), rather than a single Horton set. 
Normally, this does not affect the properties that make Horton sets notable. For example, none of the them
have empty heptagons. In some circumstances it does; as is the case of the constructions 
with few $k$-holes \cite{valtr95,dumi00,valtr04}. We fixed one of 
these two options in order to make the proof of our lower bound more readable, but our results should
hold for the general setting. Note that this choice fixes the order type of the Horton
set. However, an arbitrary drawing of the Horton set need not satisfy Definition~\ref{def:mat}.

Horton described his set in a concrete manner with specific integer 
coordinates. Another description given in~\cite{emptysimplices}, is the following.

\begin{itemize}
  \item $H^0:= \{(1,1)\}$.
  \item $H^1:=\{(1,1),(2,2)\}$.
  \item $H^k:=\{(2x-1,y):(x,y) \in H^{k-1} \} \cup \{(2x,y+3^{2^{k-1}}):(x,y) \in H^{k-1}\}$.
\end{itemize}

This drawing and the original due to Horton have exponential size; we have not seen
in the literature a drawing of subexponential size. Then again, to the
best of our knowledge nobody has tried to find small drawings of the Horton set.

 \section{Upper bound} \label{sec:upper}

In this section we construct a small drawing of 
the Horton set of $n:=2^k$ points. First, we define
the following two functions.

\begin{align*}
f(i)= & \begin{cases}
0 & \,\mbox{if}\, i=1.\\
2^{\frac{i(i-1)}{2}-1} & \,\mbox{if}\, i\geq2.
\end{cases}\\
g(i)= & \begin{cases}
0 & \,\mbox{if}\, i=1.\\
f(i)-f(i-1) & \,\mbox{if}\, i\geq2.
\end{cases}
\end{align*}

Afterwards, we use $f$ and $g$ to construct our drawing recursively as follows.

\begin{itemize}
 \item $P^0:=\{(0,0)\}$;
 
 \item $P_{\textrm{even}}^i:= \{(2x,y):(x,y)\in P^{i-1} \}$; 
 
 \item $P_{\textrm{odd}}^i:=\{ (2x+1,y+g(i)):(x,y)\in P^{i-1}\}$;
 
 \item $P^i:=P_{\textrm{even}}^i \cup P_{\textrm{odd}}^i$.
 
\end{itemize}

\begin{figure}
\begin{center}
\includegraphics[width=0.7\textwidth]{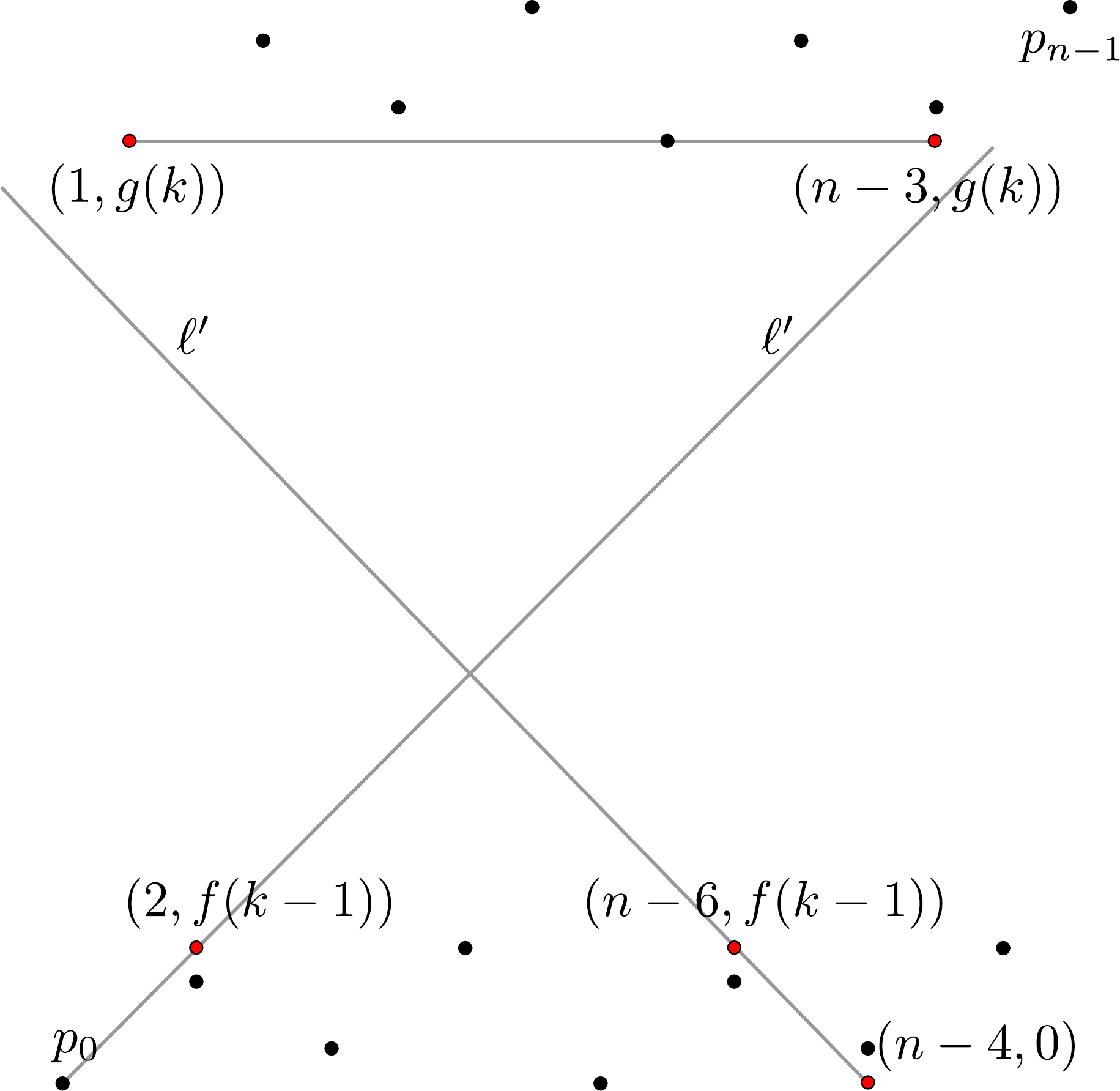}
\end{center}
\caption{The two possible definitions of $\ell'$ in the proof of Theorem~\ref{thm:upper}.}
\label{fig:horton_upper}
\end{figure}

\begin{theorem}
\label{thm:upper}
There exist a drawing of the Horton set of $n$ points of size $\frac{1}{2} n^{\frac{1}{2} \log (n/2)}$
for $n \ge 16$.
\end{theorem}
\begin{proof}
We prove by induction on $k$ that $P^k$ is the desired drawing.  
It can be verified by hand that $P^4$ has size equal to $32=\frac{1}{2} 16^{\frac{1}{2} \log (16/2)}$; 
assume that $k \ge 5$.
By induction $P_{\textrm{even}}^{k}$ and $P_{\textrm{odd}}^k$ are Horton sets;
it only remains to show that $P_{\textrm{odd}}^k$ is high above $P_{\textrm{even}}^k$.
We only prove that every point of $P_{\textrm{odd}}^k$ is above
every line through two points of $P_{\textrm{even}}^k$;
the proof that every point of $P_{\textrm{even}}^k$ is below
every line through two points of $P_{\textrm{odd}}^k$ is analogous.

Let $p_0, p_1,\dots, p_{n-1}$ be the points of $P^k$ sorted by their
$x$-coordinate. 
Let $0\le i<j \le n-1$ be two even integers, and 
let $\ell$ be the directed line from $p_i$ to $p_j$.  
By definition $P_{\textrm{odd}}^{k}$ is above the vertical line passing through 
$p_1$; in particular $P_{\textrm{odd}}^k\setminus\{ p_{n-1}\}$ is above the line segment
joining $p_1$ and $p_{n-3}$. Since the smallest $y$-coordinate
of $P_{\textrm{odd}}^k$ is equal to $g(k)$,  $p_1$ and $p_{n-3}$ are above the line
segment joining the points $(1,g(k))$ and  $(n-3,g(k))$. 
Therefore, it suffices to show that $(1,g(k))$, $(n-3,g(k))$ and $p_{n-1}$
are above $\ell$. 

We define a line $\ell'$ with the property that if $(1,g(k))$ and  $(n-3,g(k))$
are above $\ell'$, then $(1,g(k))$, $(n-3,g(k))$ and $p_{n-1}$ are above $\ell$.
Afterwards we show that indeed $(1,g(k))$ and  $(n-3,g(k))$ are above $\ell'$.

If the slope of $\ell$ is non-positive, define 
$\ell'$ to be the line passing through the points $(n-6,f(k-1))$ and $(n-4,0)$;
if the slope of $\ell$ is positive, define $\ell'$ to 
be the line passing through the points $(0,0)$ and $(2,f(k-1))$. 
Note that the largest $y$-coordinate of $P_{\textrm{even}}^k$ 
is equal to $\sum_{i=1}^{k-1} g(i)=f(k-1)$. Therefore
the slope of $\ell$ is at least $-f(k-1)/2$ and at most $f(k-1)/2$;
in particular the absolute value of the slope of $\ell'$ is larger or equal
to the absolute value of the slope of $\ell$. 
The farthest point of $P_{\textrm{even}}^k$ to the right that $\ell$ can contain while having non-positive
slope is $p_{n-4}$ (which has $x$-coordinate equal to $n-4$); the farthest point of $P_{\textrm{even}}^k$ to 
the left that $\ell$ can contain while having positive
slope is $p_0$. Therefore in both cases if $(1,g(k))$ and  $(n-3,g(k))$
are above $\ell'$, then they are also above $\ell$; see Figure~\ref{fig:horton_upper}.

If $\ell$ has non-positive
slope and $(1,g(k))$ is above $\ell$, then $p_{n-1}$ is also above
$\ell'$ since $p_{n-1}$ has larger $x$-coordinate. If $\ell$ has positive
slope and $(n-3,g(k))$ is above $\ell'$, then $p_{n-1}$ must
also be above $\ell$. Otherwise $\ell$ intersects the line segment
joining $(n-3,g(k))$ and $p_{n-1}$; this line segment
has slope equal to $f(k-1)/2$, since the $y$-coordinate of $p_{n-1}$ is
equal to $\sum_{i=1}^{k} g(i)=f(k)$. This in turn would imply that
$\ell$ has slope larger than $f(k-1)/2$---a contradiction.

Suppose $\ell$ has non-positive slope. Then it suffices
to show that $(1,g(k))$ is above $\ell'$. This is the case
since:

\begin{eqnarray*}
 &  & \hspace{-2em}\left|\begin{array}{ccc}
n-6 & f(k-1) & 1\\
n-4 & 0 & 1\\
1 & g(k) & 1
\end{array}\right|\\
 & = & 2g(k)-(n-5)f(k-1)\\
 & = & 2f(k)-(n-3)f(k-1)\\
 & = & 2f(k)-2^kf(k-1)+3f(k-1)\\
 & = & 3f(k-1)\\
 & > & 0.
\end{eqnarray*}

Suppose $\ell$ has positive slope. Then it suffices
to show that $(n-3,g(k))$ is above $\ell'$. This is the case
since:

\begin{eqnarray*}
 &  & \hspace{-2em}\left|\begin{array}{ccc}
0 & 0 & 1\\
2 & f(k-1) & 1\\
n-3 & g(k) & 1
\end{array}\right|\\
 & = & 2g(k)-(n-3)f(k-1)\\
 & = & 2f(k)-(n-1)f(k-1)\\
 & = & 2f(k)-2^kf(k-1)+f(k-1)\\
 & = & f(k-1)\\
 & > & 0.
\end{eqnarray*}

Finally the largest $x$-coordinate of $P^k$ is equal to $n-1$, and the
largest $y$-coordinate of $P^k$ is equal to 
\[ \sum_{i=1}^{k} g(i)=f(k)=2^{\frac{k(k-1)}{2}-1}=\frac{1}{2} n^{\frac{1}{2} \log (n/2)},\]
since $k=\log n$.
Therefore, $P^k$ is a drawing of the Horton set of $n$ points of size $\frac{1}{2} n^{\frac{1}{2} \log (n/2)}$.
\end{proof}


\section{Lower bound} \label{sec:lower}

In this section we prove a lower bound on the size of any
drawing of the Horton set. 
As mentioned before, a drawing of the Horton set might not
satisfy Definition~\ref{def:mat};
we call a drawing  that does, an \emph{isothetic} drawing
of the Horton set.
We first show a lower bound on the size of isothetic drawings
of the Horton set (Theorem~\ref{thm:lower_isothetic}); afterwards, we 
consider the general case (Theorem~\ref{thm:lower_gen}).
Throughout this section $P$ is an isothetic drawing
of the Horton set of $n:=2^k$ points, and $p_0, p_1,\dots, p_{n-1}$
are the points of $P$ sorted by their $x$-coordinate. 

As an auxiliary structure, we recursively define a complete rooted binary
tree $T$, as follows.  $P$ is the root of $T$; 
and if $Q \subset P$ is a vertex of $T$, of at least two points, 
then $Q_{\textrm{even}}$ and $Q_{\textrm{odd}}$ are its left and right children,
respectively. Furthermore, for each vertex in $T$, label the edge incident to its left
child with a ``0'' and the edge incident with its right child with  a``1'';
the labels encountered in a path from a leaf $\{p_i\}$ to the root are
precisely the bits in the binary expansion of $i$; see Figure~\ref{fig:T}.

By construction, the vertices of $T$ are sets of $2^i$ points of $P$ (for some $0 \le i \le k$). Let $T_i$ be the 
set of vertices of $T$
that consist of exactly $2^i$ points of $P$: we call it the $i$-th\footnote{In the literature
the $i$-th level of a binary tree are those vertices at distance $i$ from the root; we have precisely
the opposite order.} level of $T$. The first level, $T_1$, are the vertices of $T$ that
consist of a pair of points of $P$. For each such pair, we consider the line through them as defined
by them. 

Let $R$ be the closed vertical slab bounded by the vertical lines
through $p_{n/4}$ and $p_{3n/4-1}$. Let $Q$ be a vertex at the first level of $T$
and let $p_i$ and $p_j$ be its leftmost and rightmost points
respectively.
Suppose that $Q$ is a left child. Then the two most significant bits in the 
binary expansion of $i$ are ``00'',
and the two most significant bits in the binary
expansion of $j$ are ``10''.
This implies that $i \le n/4$ and $j-i=n/2$; in particular
$p_j$ is contained in $R$, while $p_i$ is to the left of $R$.
In this case, we say that $Q$ is \emph{left-to-right} crossing.
By similar arguments if $Q$ is a right child, then
$p_i$ is contained in $R$, while $p_j$ is to the right of $R$.
In this case we say that $Q$ is \emph{right-to-left} crossing.
Note that the vertices in the first level of $T$, in their left to right
order in $T$, are alternatively \emph{left-to-right} and 
\emph{right-to-left} crossing (see Figure~\ref{fig:T}).
The following lemma relates the left to right order of these vertices in $T$,
to the bottom-up order of their corresponding pairs of points of $P$.

\begin{figure}
\begin{center}
\includegraphics[width=1.0\textwidth]{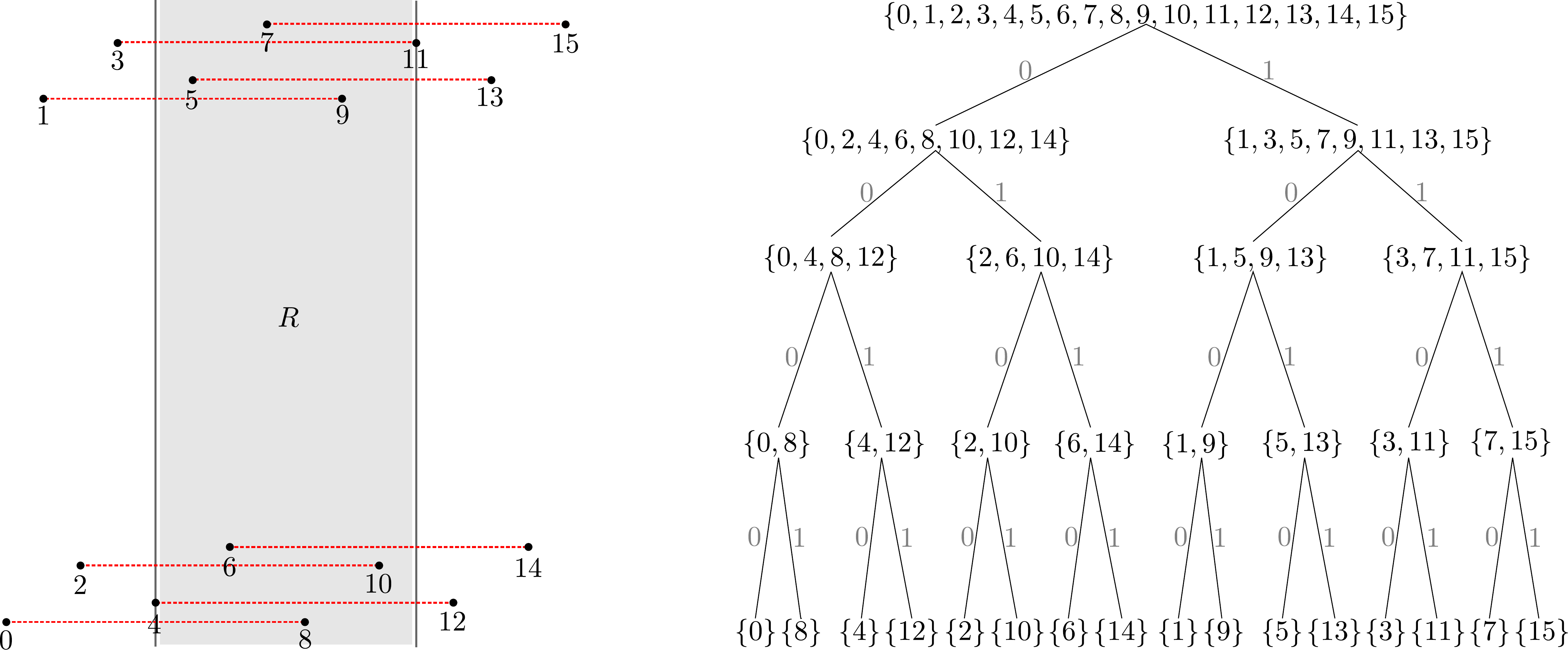}
\end{center}
\caption{The Horton set and its associated tree $T$.}
\label{fig:T}
\end{figure}

\begin{lemma}\label{lem:no_cut} 
The lines defined by the vertices
of the first level of $T$ do not intersect inside $R$. In particular, the bottom-up
order of convex hull of these vertices corresponds to their
left to right order in $T$.
\end{lemma}
\begin{proof}

Let $Q_1$ and $Q_2$ be two vertices in the first level of $T$
such that $Q_1$ is a left child and $Q_2$ is a right child.
Without loss of generality assume that in the left to right order in $T$, $Q_1$ 
is before $Q_2$. 
Then, $Q_1$ is left-to-right crossing and $Q_2$ is
right-to-left crossing. Let $\gamma_1$ and $\gamma_2$ 
be the lines defined by $Q_1$ and $Q_2$, respectively.
If $\gamma_1$ and $\gamma_2$ intersect inside $R$, then
the leftmost point of $Q_1$ is above $\gamma_2$ or
the rightmost point of $Q_2$ is below $\gamma_1$---a contradiction to
property 3 of Definition~\ref{def:mat}.
Since between every two left children
there is a right child, and between every two right children
there is a left child, the result follows.
\end{proof}

By construction every vertex of $T$ is an isothetic drawing of the Horton set.  
The main idea behind the proof of the lower bound on the size of isothetic drawings
of the Horton set is to lower bound the size of these
drawings in terms of the size of their children. We define some parameters
on the vertices of $T$, that make this idea more precise. 

Let $2 \le t \le k$ be an integer. Let $\ell_1, \ell_2, \ell_3$ and $\ell_4$ be four vertical lines sorted from left to right, such that:
\begin{itemize}
 \item All of them are contained in the interior of $R$.
 
 \item There are exactly $2^{k-t}$
points of $P$ between both pairs ($\ell_1,\ell_2$) and $(\ell_3,\ell_4$). 
\end{itemize}

Let $Q$ be a vertex of $T$ with more than two points.
For each of the $\ell_i$,
we define two parameters of $Q$.
Let $\gamma_D(Q)$ be the line defined by the leftmost descendant
of $Q$ in $T_1$. Let $\gamma_U(Q)$ be the line defined by the rightmost descendant
of $Q$ in $T_1$. Note that $Q$ is bounded from below by $\gamma_D(Q)$ and from above by
$\gamma_U(Q)$ (Lemma~\ref{lem:no_cut}). Let $Q_L$ and $Q_R$ be the left and right children of $Q$, respectively.
Define $\operatorname{width}_i(Q)$ as the distance
between the points $\gamma_D(Q) \cap \ell_i$ and $\gamma_U(Q) \cap \ell_i$, and 
$\operatorname{girth}_i(Q)$ as the distance between the points 
$\gamma_U(Q_L) \cap \ell_i$ and $\gamma_D(Q_R) \cap \ell_i$; see Figure~\ref{fig:width_girth}. 

\begin{figure}
   \centering
\includegraphics[width=0.7\textwidth]{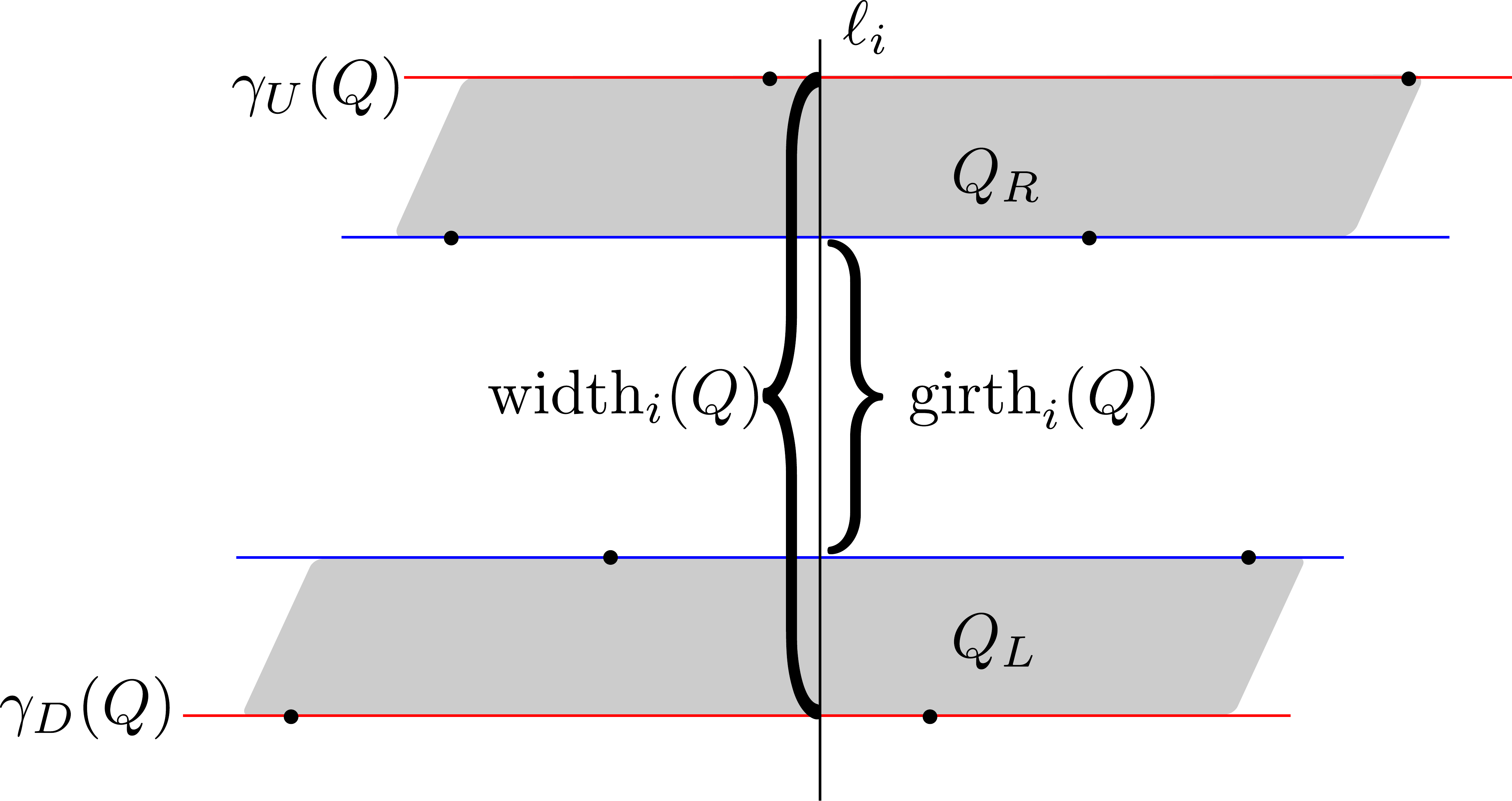}
\caption{The bounding lines of $Q$, together with its width and girth.}
	\label{fig:width_girth}
\end{figure}

We lower bound the girth of a vertex of $T$ in terms
of the girth of one of its children. This bound is expressed in Lemma~\ref{lem:lower_bound}.
Before proceeding we need one more definition. Let $Q$ be a vertex of $T$ with more than two points and let $P(Q)$ be its parent. 
If $Q$ is the left child of $P(Q)$, let $S(Q)$ be the right child of $Q$; otherwise let $S(Q)$ be the
left child of $Q$.

\begin{figure}
\begin{center}
\includegraphics[width=0.5\textwidth]{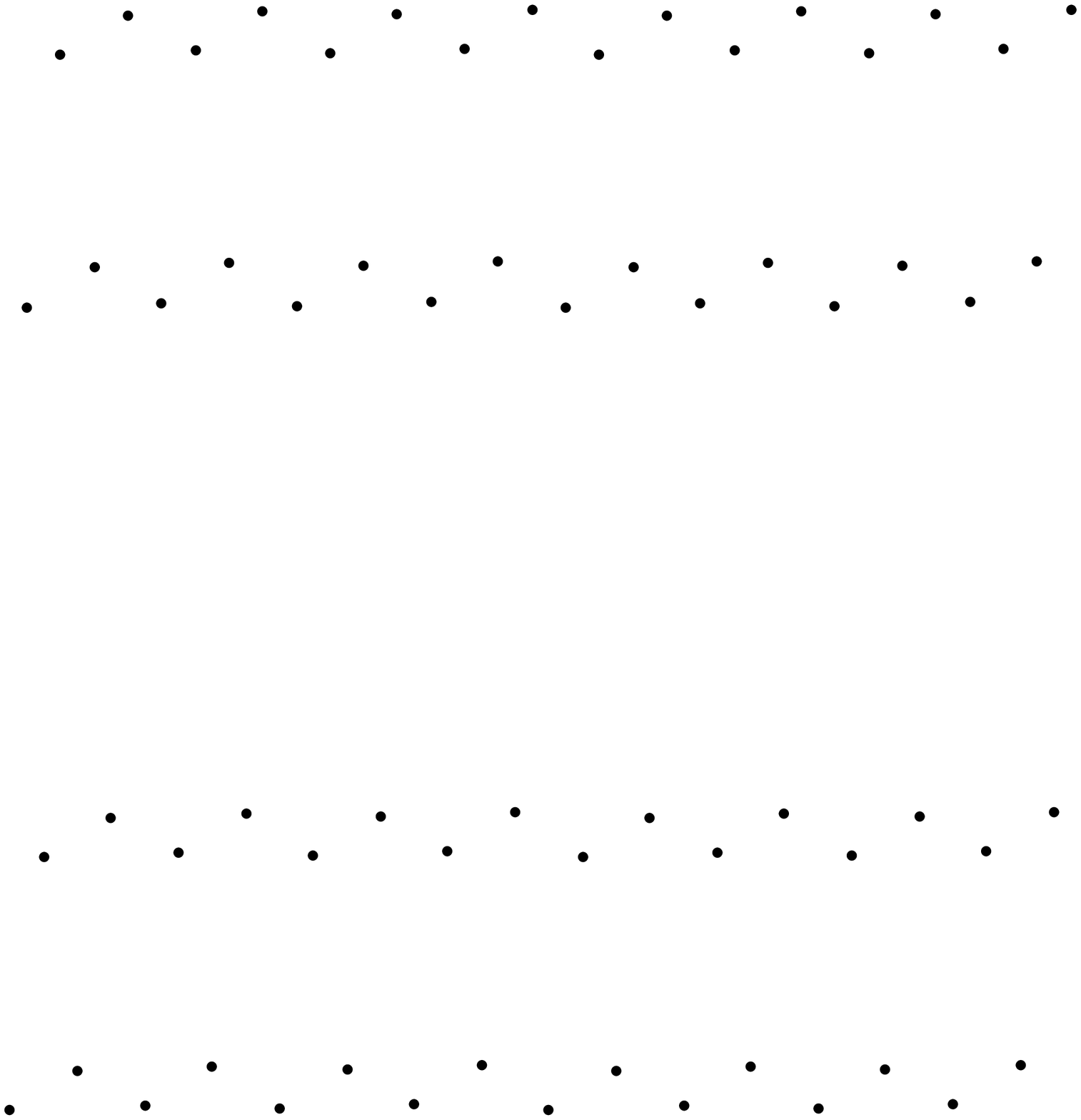}
\end{center}
\caption{Schematic depiction of the proof of Lemma~\ref{lem:lower_bound}.}
\label{fig:push}
\end{figure}

\begin{lemma}
\label{lem:lower_bound}
  Let $Q$ be a vertex at the $l$-th level of $T$, for some $t< l< k$.
   If the distance between $\ell_1$ and
  $\ell_{2}$ is $d_1$, and the distance between $\ell_{3}$ and
  $\ell_{4}$ is $d_2$, then:

\begin{itemize}
  \item[(1)] $\operatorname{girth}_1(P(Q)) \geq  \left( \frac{(d_1)^2}{(d_1+d_2)d_2} \right ) 2^{l-t-1} \operatorname{girth}_4(Q)-\operatorname{width}_1(S(Q))$ and,
  
  \item[(2)] $\operatorname{girth}_4(P(Q)) \geq  \left( \frac{(d_2)^2}{(d_1+d_2)d_1} \right ) 2^{l-t-1} \operatorname{girth}_1(Q)-\operatorname{width}_4(S(Q))$.

\end{itemize}

\end{lemma}

\begin{proof}
We will prove inequality $(1)$; the proof of $(2)$ is analogous. 
Assume that $Q$ is the left child of $P(Q)$ and let $Q'$ be the right child of $P(Q)$; the case when $Q$ is the right child of $P(Q)$
can be proven with similar arguments. Note that is $S(Q)$ is the right child, $Q_R$, of $Q$.
 
Let $p_1'$ and $p_2'$ be two consecutive points in $Q_L$ lying between $\ell_{3}$
and $\ell_4$ at a horizontal distance of at most  
 $\Delta_x:=d_{2}/2^{l-t-1}$ from each other; such a pair exists as there  are $2^{l-t-1}$ points
of $Q_L$ between $\ell_{3}$ and $\ell_4$.
Let $p''$ be the point in $Q_R$ that lies between $p_1'$ and $p_2'$ (in the $x$-coordinate order).
Let $\varphi$ be the line through $p_2'$ and $p''$. 
Let $\Delta_y:=\min\left\{\operatorname{girth}_{3}(Q),\operatorname{girth}_4(Q)\right\}$; 
note that the slope of $\varphi$ is at most  $-\Delta_y/\Delta_x$.
Recall that by Lemma~\ref{lem:no_cut}, $\gamma_D(Q_R)$ and $\gamma_U(Q_L)$ do not
intersect between $\ell_1$ and $\ell_4$; this implies that $\operatorname{girth}_{3}(Q) \ge \frac{d_1}{d_1+d_2}\operatorname{girth}_{4}(Q)$,
in particular
 $\Delta_y \ge \frac{d_1}{d_1+d_2}\operatorname{girth}_{4}(Q)$.
Therefore, the slope of $\varphi$ is at most
\[-\Delta_y/\Delta_x=-\left ( \frac{d_1}{d_1+d_2}\operatorname{girth}_{4}(Q) \right )/\Delta_x=\frac{d_1}{(d_1+d_2)d_2}2^{l-t-1}\operatorname{girth}_{4}(Q).\]

Define the following points $q_1:=\gamma_D(Q_R) \cap \ell_{1}$,  $q_2:=\varphi \cap \ell_1$ 
and $q_3:=\gamma_D(Q') \cap \ell_{1}$ (see Figure~\ref{fig:push}).
Note that the leftmost
point of $\gamma_D(Q')\cap Q'$ is to the left of $\ell_1$; since this point
is above $\varphi$, $q_2$ cannot be above $q_3$.
Therefore, the distance from $q_1$ to $q_2$
is at most the distance from $q_1$ to $q_3$; the distance from $q_1$ to $q_3$ is precisely 
$\operatorname{girth}_1(P(Q))+ \operatorname{width}_1(S(Q))$. We now show that the distance
from $q_1$ to $q_2$ is at least 
$\frac{(d_1)^2}{(d_1+d_2)d_2} 2^{l-t-1} \operatorname{girth}_4(Q)$---this completes the proof of $(1)$.

Let $\varphi'$ be the line parallel to $\varphi$ and passing through the intersection
point of $\ell_{3}$ and $\gamma_D(Q_R)$. Note that $\varphi'$ is below
$\varphi$. Therefore, the distance from $q_1$ to $q_2$ is at least the
distance of $q_1$ to the intersection point 
of $\varphi'$: this is at least $d_1(\Delta_y/\Delta_x)=\frac{(d_1)^2}{(d_1+d_2)d_2} 2^{l-t-1} \operatorname{girth}_4(Q)$.
\end{proof}  

Two obstacles prevent us from directly applying Lemma~\ref{lem:lower_bound}. One is that the difference
between $d_1$ and $d_2$  may be too big and in consequence $\frac{(d_1)^2}{(d_1+d_2)d_2}$
or $\frac{(d_2)^2}{(d_1+d_2)d_1}$ too small. This situation can be fixed with following Lemma.

\begin{lemma}\label{lem:ratio}
 For $t:=\lceil 2\log k\rceil$ and  $k\ge 16$, $P$ has size at least $n^{\frac{1}{2} \log n}$ or $\ell_1, \ell_2, \ell_3, \ell_4$ can 
 be chosen so that the ratio between $d_1$ and $d_2$ is at least $1/2$
 and at most $2$.
\end{lemma}
\begin{proof}
Let $\varphi_1, \dots, \varphi_{2^{t-1}}$ be consecutive vertical lines such that:

\begin{itemize}
\item all of them lie in the interior of $R$ and,

\item between every pair of two consecutive lines $(\varphi_i, \varphi_{i+1})$ there are exactly $2^{k-t}$ points of $P$.
\end{itemize}

For $1 \le i < 2^{t-1}$, let $\Delta_i$ be the distance between $\varphi_i$ and  $\varphi_{i+1}$;
let $\Delta_1' \le \Delta_2'\le \dots \le  \Delta_{2^{t-1}-1}'$ be these distances sorted by size. 
We look for a pair $\Delta_i' \le \Delta_j'$, such that one is at most two times the other.
Suppose there is no such pair; then $\Delta_{i+1}'\ge 2\Delta_i'$. Since between the two lines
defining $\Delta_1'$ there are exactly $2^{k-t}$ points of $P$, and no three
of them have the same integer $x$-coordinate, $\Delta_1' \ge 2^{k-t-1}$.
Therefore,
\[\Delta_{2^{t-1}-1}'\ge  2^{k-t-1}\cdot 2^{2^{t-1}-2}\ge 2^{\frac{1}{2}k^2+k-t-3} \ge 2^{\frac{1}{2}k^2+k-2\log(k)-4} \ge n^{\frac{1}{2}\log n}. \]
The latter part of the inequality follows from our assumption that $k \ge 16$. 
Therefore, if there is no such pair, $P$ has size at least  $n^{\frac{1}{2} \log n}$.
 
\end{proof}


The second obstacle
is that the second term in the right hand sides of inequalities (1) and (2)
of Lemma~\ref{lem:lower_bound} may be too large. In this case, we  prune
$T$ to get rid of vertices of large width; this is done by choosing 
an integer $l \le k-1$ and then removing from
$P$ all the points are contained in either: all the vertices
of $T_l$ that are a left child to their parent, or all the 
the vertices of $T_l$ that are a right child to their parent 
(see Figure~\ref{fig:T_pruned}). We call this operation 
\emph{pruning} the $l$-th level of $T$. The resulting set is a drawing
of the Horton set, as shown by the following lemma.

\begin{figure}
\begin{center}
\includegraphics[width=1.0\textwidth]{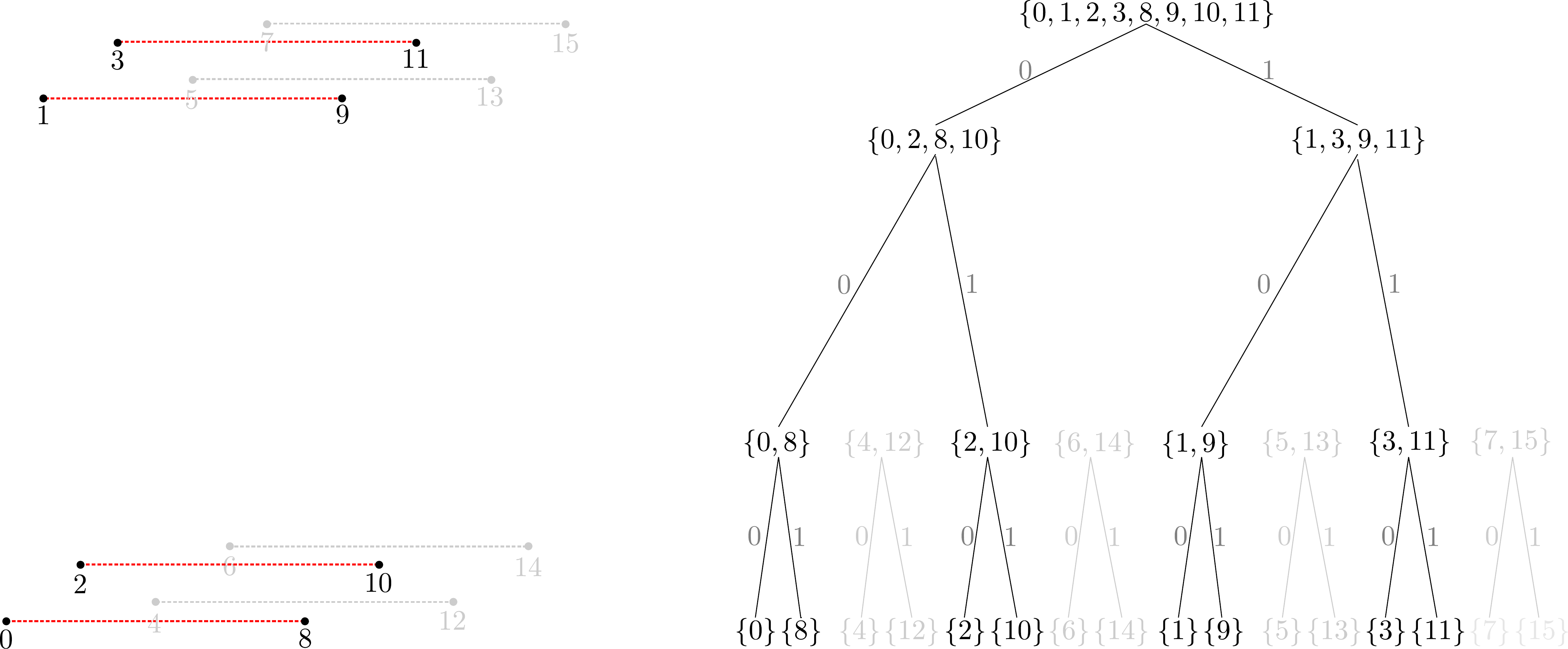}
\end{center}
\caption{$P$ and $T$ after the removal of the vertices of $T_1$ that 
are a right child to their parent.}
\label{fig:T_pruned}
\end{figure}

\begin{lemma}\label{lem:prune}
Let $P'$ be the subset of $P$ that results from pruning the $l$-th level of $T$.
Then:

\begin{itemize}
 \item[(1)]$P'$ is an isothetic drawing of the Horton set of $n/2$ points.
 
 \item[(2)] Suppose that $l \le k-3$.  Let $T'$ be the tree associated to $P'$, and $Q'$ be any vertex at the $l$-th level of $T'$ (for
 some $l'>l$).
 Then there exist a vertex
  $Q$  at the $(l'+1)$-level of $T$ that contains $Q'$. Moreover, $S(Q')\subset S(Q)$.
\end{itemize}
\end{lemma}
\begin{proof}

Assume without loss of generality that the left children are removed when pruning
$T$.
If $l=k-1$, (2) holds trivially, and (1) holds because in that case $P'=P_{\textrm{odd}}$.
Assume that $l\le k-2$, and let $s:=k-l$; we proceed by induction on $s$.

Note that  $P_{\textrm{even}}$ and $P_{\textrm{odd}}$ are each an
isothetic drawing of the Horton set of $n/2$ points. Moreover, their corresponding
trees, $T_{\textrm{even}}$ and $T_{\textrm{odd}}$, are the subtrees of $T$ rooted at  $P_{\textrm{even}}$ and $P_{\textrm{odd}}$,
respectively. Therefore, when we prune the $l$-th level of $T$, we also prune the $l$-th level of $T_{\textrm{even}}$ and $T_{\textrm{odd}}$.
By induction and (1), this produces two isothetic drawings of the Horton set
of $n/4$ points; let $P_0''\subset P_{\textrm{even}} $ and $P_1'' \subset P_{\textrm{odd}} $ be these drawings, respectively.

We first prove that 

\MyQuote{$P'$ can be constructed from $P$ by, starting at $p_0$, alternatively removing
and keeping intervals of $2^{k-l-1}$ consecutive points of $P$.}

For $s=1$, this is
trivial since  $2^{k-l-1}=1$ and $P'=P_{\textrm{odd}}$. Thus, by induction,
$P_0''$ and $P_1''$ are constructed
from $P_{\textrm{even}}$ and $P_{\textrm{odd}}$ by, starting at their leftmost point,
alternatively removing and keeping intervals of $2^{k-l-2}$ consecutive
points of $P_{\textrm{even}}$ and $P_{\textrm{odd}}$, respectively. 
Let $I_1', \dots, I_{2^{l+1}}' \subset P_{\textrm{even}}$   and 
$J_1', \dots, J_{2^{l+1}}' \subset P_{\textrm{odd}}$ be these
intervals (in order). Finally, $(\ast)$ follows by letting $I_i:=I_i' \cup J_i'$.

We now prove \emph{1} and \emph{2}.
\begin{itemize}
 \item[\emph{(1)}] 
Note that $(\ast)$ implies that $P_{\textrm{even}}'=P_0''$ and $P_{\textrm{odd}}'=P_1''$.
Thus $P_{\textrm{even}}' \subset P_{\textrm{even}}$ and
$P_{\textrm{odd}}' \subset P_{\textrm{odd}}$; in particular $P_{\textrm{odd}}'$ is high above $P_{\textrm{even}}'$.
Therefore, $P'$ is an isothetic drawing of the Horton set of $n/2$ points.

\item[\emph{(2)}]

Consider the following algorithm. Remove from $T$ the subtrees rooted at the vertices in the $l$-th level of $T$ that
are a left child to their parent; afterwards, remove from each vertex of $T$ the points in $P \setminus P'$. After this last 
  step, each vertex at the $l$-th level of $T$  that was not removed is equal to its parent---producing a loop;
  remove these loops. We claim that this algorithm produces $T'$. For $s=1$, this follows from $(\ast)$. 
  Let  $T_{\textrm{even}}'$ and $T_{\textrm{odd}}'$ be the left and right
  subtrees of the root of $T'$, respectively.
  By induction $T_{\textrm{even}}'$ and $T_{\textrm{odd}}'$ can  constructed from
  $T_{\textrm{even}}$ and $T_{\textrm{odd}}$ with the above algorithm, respectively.
  Since the root of $T'$ is precisely the root of $T$ minus the points in $P \setminus P'$,
  the algorithm produces $T'$.
  
  Now, suppose that $l \le k-3$ and let $Q'$ be a vertex at the $l'$-th level of $T'$, for
  some $l'>l$. By
  the algorithm, there is a vertex $Q$ such that $Q'=Q\setminus (P \setminus P')$;
  this vertex is in the $(l'+1)$-level of $T$. Finally, also by the algorithm
  we have that $S(Q') \subset S(Q)$.
\end{itemize}
\end{proof}

We are now ready to prove our lower bound on the size of isothetic drawings of the Horton set.

\begin{theorem}\label{thm:lower_isothetic}
For a sufficiently large value of $k$, every isothetic drawing of the Horton set of $n=2^k$ points has size
at least $n^{\frac{1}{8} \log n}$.
\end{theorem}
\begin{proof}
Set $t:=\lceil 2\log k\rceil$ and assume that $k\ge 16$. By Lemma~\ref{lem:ratio} 
$\ell_1, \ell_2, \ell_3$ and $\ell_4$ can be chosen so that, the ratio of the distance
between $d_1$ and $d_2$ is at least $1/2$ 
and at most $2$.  Without loss of generality assume that $d_1 \le d_2$.
Let $D$ be the distance between $\ell_1$ and $\ell_4$. We may assume that
\[D< n^{\frac{1}{8}\log n};\] as otherwise we are done.

Let $Q$ be a vertex in the $(t+1)$-th level of $T$. 
Note that between two consecutive points in every vertex at the $l$-th level
of $T$ there are exactly $2^{k-l}-1$ points of $P$. This trivially holds
for  $l=k$; it holds for smaller values of $l$, by induction
on $k-l$. In particular,
there are $(2^{t+1}-1)(2^{k-t-1}-1)+2^{t+1}-2=2^k-2^{k-t-1}-1$ points
of $P$ between the leftmost and rightmost point of $Q$.

This implies that 
there are exactly two points of $Q$ between $\ell_1$ and $\ell_{2}$,
and exactly two points of $Q$ between $\ell_{3}$ and $\ell_4$. 

Suppose that there
were less than two points of $Q$ between $\ell_1$ and $\ell_{2}$, then the number of points of $P$ would
be at least the sum of the following. 
\begin{itemize}
 \item The number of points of $P$ between $\ell_1$ and $\ell_{2}$; recall that this is
 equal to $2^{k-t}$.

\item The number
of points of $P$ between the leftmost and rightmost point of $Q$ that are not between  
$\ell_1$ and $\ell_{2}$;  since there are exactly $2^{k-t-1}-1$ points of $P$ between two consecutive points of $Q$,
and at most one point of $Q$ between $\ell_1$ and $\ell_2$, this is at least $(2^k-2^{k-t-1}-1)-2^{k-t-1}=2^k-2^{k-t}$.

\item Two, for the leftmost and rightmost point of $Q$.
\end{itemize}

In total this is at least $2^k+1=n+1$---a contradiction; similar arguments hold for $\ell_3$ and $\ell_4$.

Suppose that there are more than two points of $Q$ between $\ell_1$ and 
$\ell_2$, then the number of points of $P$ between $\ell_1$ and $\ell_2$ is at least
$2(2^{k-t-1})+3=2^{k-t}+3$; this is a contradiction to the assumption
that there are exactly $2^{k-t}$ points of $P$ between $\ell_1$
and $\ell_2$. The same argument holds for $\ell_3$ and $\ell_4$. 

The two points of $Q$  between $\ell_1$ and $\ell_2$, and the two points 
of $Q$  between $\ell_3$ and $\ell_4$, have integer coordinates. Therefore,
by Pick's theorem~\cite{pick} the area of their convex hull is at least
one. Since these points are contained in trapezoid bounded by $\gamma_D(Q)$, $\gamma_U(Q)$, $\ell_1$
and $\ell_4$, the area of this trapezoid is also at least one. But this area is at most 
$D(\operatorname{width}_1(Q)+\operatorname{width}_4(Q))/2$. Therefore

\begin{equation} \label{eq:w>=1}
 \max\{\operatorname{width}_1(Q),\operatorname{width}_4(Q)\} \ge 1/D.
\end{equation}

This bound also holds for every vertex at a level higher than $t+1$ 
(since all of these vertices contain vertices at the $t+1$-level as subsets).

Let $t < l \le k$ be the largest positive integer  such that
there exists a vertex $R$ in the $l$-th level of $T$ that satisfies:

\begin{equation}\label{eq:s(r)}
\max\{\operatorname{width}_1(S(R)),\operatorname{width}_4(S(R))\} \ge \frac{2^{(l-t-6)(l-t-7)/2}}{D}.
\end{equation}

Such an $l$ and $R$ exist since (\ref{eq:s(r)}) holds 
for every vertex at the $(t+6)$-th level of $T$. Indeed
if $Q$ is a vertex at the $(t+6)$-th level of $T$, then 
$S(Q)$ is in the $(t+5)$ level of $T$ and by (\ref{eq:w>=1}):

\[\max\{\operatorname{width}_1(S(Q)),\operatorname{width}_4(S(Q))\} \ge 1/D=\frac{2^{((t+6)-t-6)((t+6)-t-7)/2}}{D}.\]

Without loss of generality assume that 
$\operatorname{width}_1(S(R)) \ge (2^{(l-t-6)(l-t-7)/2})/D$ and that $R$ is a left
child.
We may assume that $l<k$, otherwise (\ref{eq:s(r)}) implies that $P$ has size at least 
$n^{\frac{1}{8}\log n}$ (for a sufficiently large value of $k$). 

We now apply Lemma~\ref{lem:prune} to prune $T$ of all the vertices
of large width (that, is that satisfy (\ref{eq:s(r)})).
Prune the $l$-th level of $T$ by removing all the vertices that are a left
child to their parent. Let $P'$ be the resulting point set
and $T'$ its corresponding tree. No vertex of $T'$ in a level higher than $l$ satisfies (\ref{eq:s(r)});
otherwise, by part $(2)$ of Lemma~\ref{lem:prune} there would be a vertex
at level of $T$ higher than $l$ that satisfies (\ref{eq:s(r)}). 

Let $(P(R)'=Q_l', Q_{l+1}',\dots, Q_{k-1}'=P')$ be the path from $P(R)'$
to the root of $T'$. We prove inductively for $l \le m \le k-1$, that:
\begin{align}
  \operatorname{girth}_{1}(Q_m') & \ge \frac{2^{(m-t-6)(m-t-7)/2}}{D} \, \mbox{if} \, m \equiv l \mod 2, \label{eq:g1}\\ 
  \operatorname{girth}_{4}(Q_m') & \ge \frac{2^{(m-t-6)(m-t-7)/2}}{D} \, \mbox{if} \, m \not \equiv l \mod 2 \label{eq:g4}
\end{align}
(\ref{eq:g1}) holds for $m=l$ since 
$\operatorname{girth}_1(Q_{l+1}')= \operatorname{girth}_1(P(R)')\ge \operatorname{width}_1 (S(R)) \ge(2^{(l-t-6)(l-t-7)/2})/D$.
Assume that $m>l$ and that both (\ref{eq:g1}) and (\ref{eq:g4}) hold for smaller values of $m$.
Suppose that $m$ has the same parity as $l$.
Then by inequality (1) of Lemma~\ref{lem:lower_bound} and inequalities (\ref{eq:s(r)}) and (\ref{eq:g1}):

\begin{align*}
 \operatorname{girth}_{1}(Q_m') &\ge  \left( \frac{(d_1)^2}{(d_1+d_2)d_2} \right ) 2^{m-t-2} 
 \operatorname{girth}_4(Q_{m-1}')-\operatorname{width}_1(S(Q_{m-1}'))\\
&\ge 2^{m-t-5} \operatorname{girth}_4(Q_{m-1}')-\frac{2^{(m-t-7)(m-t-8)/2}}{D}\\
&\ge 2^{m-t-5}\frac{2^{(m-t-7)(m-t-8)/2}}{D}-\frac{2^{(m-t-7)(m-t-8)/2}}{D}\\
&\ge 2^{m-t-6}\frac{2^{(m-t-7)(m-t-8)/2}}{D} \\
&=\frac{2^{(m-t-6)(m-t-7)/2}}{D}
\end{align*}

Therefore $P'$ has size at least $\frac{2^{(k-t-7)(k-t-8)/2}}{D}$. This at least
$n^{\frac{1}{8}\log n}$, for a sufficiently large value of $k$. Since $P' \subset P$,
the result follows.
The proof when $m$ has different parity as $l$ is similar, but uses inequality (2) of  Lemma~\ref{lem:lower_bound}
instead.
\end{proof}

To prove the general lower bound we do the following. Take a drawing of the Horton set;
find a subset of half of its points, for which we know that there exists
a linear transformation that maps it into an isothetic drawing; afterwards, apply
Lemma~\ref{thm:lower_isothetic} to the image and use the obtained lower bound to lower bound
the size of original drawing.

\begin{theorem}\label{thm:lower_gen}
Every drawing of the Horton set of $n=2^k$ points has size at least $c \cdot n^{\frac{1}{24}\log (n/2)}$,
for a sufficiently large value of $n$ and some positive constant $c$.
\end{theorem}
\begin{proof}
Let $P'$ be a (not necessarily isothetic) drawing of the Horton set of $n$ points. 
  As $P$ and $P'$ have the same order type we can label $P'$ with the
 same labels as $P$, such that corresponding triples of points in $P$ and $P'$
 have the same orientation. Let $\{p_0',\dots,p_{n-1}'\}$ be $P'$
 with these labels.
 
 Note that the clockwise order by angle of $P'_{\textrm{odd}}$ around $p_0'$ is $(p_1',p_3',\dots)$, 
and that  $p_0'$ lies in an unbounded cell of the line arrangement of the lines defined
by every pair of points of $P_{\textrm{odd}}'$; thus, point $p_0'$  can be moved towards infinity without changing
 this radial order around $p_0'$. Therefore, there is a direction $\vec{d}$
 in which if $P_{\textrm{odd}}'$ is projected orthogonally the order of the projection
 is precisely $(p_1',p_3',\dots)$. We may rotate $\vec{d}$ as long as it does not coincide
 with a direction defined by a pair of points of $P'$ and the order of $P_{\textrm{odd}}'$
in this projection does not change. Let $v'$ and $v''$ be the first vectors,
defined by pairs of points of $P'$, encountered when rotating $\vec{d}$ to the left
and to the right, respectively; let $v=v'+v''=(a,b)$.

We may assume that $||v||=\sqrt{a^2+b^2}\le 4^{1/3}(n/2)^{\frac{1}{24}\log (n/2)}$; otherwise
one of $v'$ and $v''$ has length at least $(1/2)4^{1/3}(n/2)^{\frac{1}{24}\log (n/2)}$,
and therefore a coordinate of value at least $(1/4)^{2/3}(n/2)^{\frac{1}{24}\log (n/2)}$.
Let $v^{\perp}=(b,-a)$. Consider a change of basis from the standard
basis to $\{v,v^{\perp}\}$. Note that under this transformation $(x,y)$ is mapped
to $\left( \frac{ax+by}{a^2+b^2},\frac{ay-bx}{a^2+b^2} \right )$. We multiply the image
of $P'$ under this mapping by $a^2+b^2$, to obtain an isothetic drawing of the Horton
set on $n/2$ points. By Theorem~\ref{thm:lower_isothetic}, this drawing has size at least
$(n/2)^{\frac{1}{8}\log (n/2)}$. Therefore, $P'$ has size at least 
$((n/2)^{\frac{1}{8}\log (n/2)})/(a^2+b^2) \ge (1/4)^{2/3}(n/2)^{\frac{1}{24}\log (n/2)}$.
\end{proof}

We point out that the constants in the exponent of the lower bounds of Theorems
\ref{thm:lower_isothetic} and \ref{thm:lower_gen} can be improved. We 
simplified the exposition at the expense of these worse bounds.

\textbf{Acknowledgments.}
We thank Dolores Lara, Gustavo Sandoval and  Andr\'es Tellez for various
helpful discussions.


\small
\bibliographystyle{abbrv} \bibliography{horton}







\end{document}